\documentclass[letterpaper, 10 pt, conference]{ieeeconf}  % Comment this line out
\IEEEoverridecommandlockouts                              % This command is only
\usepackage{float,amsmath,amsfonts, color,graphicx}

\usepackage{amssymb, amsthm}
\usepackage{calc}
\usepackage{tikz}
\usepackage{pgfplots}
\usepackage{cases}
\usepackage{url}
\usepackage{algpseudocode}
\usepackage{algorithm}

\usepackage{color}

\usetikzlibrary{external}
\usetikzlibrary{decorations.markings}
\usetikzlibrary{patterns}

\newcommand{\R}{\mathbb{R}}
\DeclareMathAlphabet{\mathcal}{OMS}{cmsy}{m}{n}

\usetikzlibrary{arrows}
\usetikzlibrary{shapes}
 \usetikzlibrary{calc}
\usetikzlibrary{decorations.pathreplacing}
\usetikzlibrary{arrows,positioning} 
\usetikzlibrary{pgfplots.groupplots}

\DeclareMathOperator*{\argmin}{arg\,min}

\theoremstyle{plain}
\newtheorem{thm}{Theorem}

\newtheorem{prob_stat}{Problem}

\theoremstyle{definition}
\newtheorem{rmk}{Remark}

\newtheorem{definition}{Definition}

\newtheorem{example}{Example}

\newcommand{\C}{\mathcal{C}}
\newcommand{\E}{\mathcal{E}}

\newcommand{\K}{\mathcal{K}}
\newcommand{\norm}[1]{\left\lVert#1\right\rVert}

\newcommand{\Cross}{$\mathbin{\tikz [x=1.4ex,y=1.4ex,line width=.2ex, black] \draw (0,0) -- (1,1) (0,1) -- (1,0);}$}

\title{\LARGE \bf Extent-Compatible Control Barrier Functions}

\author{Mohit Srinivasan${}$, Matthew Abate${}$, Gustav Nilsson${}$, and Samuel Coogan%
\thanks{This work was partially supported by the National Science Foundation under Grant 1749357. The authors are with the School of Electrical and Computer Engineering, Georgia Institute of Technology, Atlanta, 30332, USA. {\tt\small \{mohit.srinivasan, matt.abate, gustav.nilsson, sam.coogan\}@gatech.edu}. M. Abate is also with the School of Mechanical Engineering and S. Coogan is also with the School of Civil and Environmental Engineering, Georgia Institute of Technology.}
}

\begin{document}

\maketitle

%%%%%%%%%%%%%%%%%%%%%%%%%%%%%%%%%%%%%%%%%%%%%%%
%% Abstract
%%%%%%%%%%%%%%%%%%%%%%%%%%%%%%%%%%%%%%%%%%%%%%%
\begin{abstract}
Safety requirements in dynamical systems are commonly  enforced with set invariance constraints over a safe region of the state space. Control barrier functions, which are Lyapunov-like functions for guaranteeing set invariance, are an effective tool to enforce such constraints and guarantee safety when the system is  represented as a point in the state space. 
In this paper, we introduce extent-compatible control barrier functions as a tool to enforce safety for the system including its volume (extent) in the physical world. In order to implement the extent-compatible control barrier functions framework, a sum-of-squares based optimization program is proposed. Since sum-of-squares programs can be computationally prohibitive, we additionally introduce a sampling based method in order to retain the computational advantage of a traditional barrier function based quadratic program controller. We prove that the proposed sampling based controller retains the guarantee for safety. Simulation and robotic implementation results are also provided.
\end{abstract}

%%%%%%%%%%%%%%%%%%%%%%%%%%%%%%%%%%%%%%%%%%%%%%%
%% Section I: Introduction
%%%%%%%%%%%%%%%%%%%%%%%%%%%%%%%%%%%%%%%%%%%%%%%
\section{Introduction}
\label{sec:intro}
A controlled dynamical system is considered \emph{safe} if it can be ensured that a given set of safe states is forward invariant under the action of a controller, i.e., the system state remains within the safe set for all time when initialized within the safe set. For example, collision avoidance between robots, obstacle avoidance during waypoint navigation, or lane changing for autonomous vehicles can be cast as invariance constraints. Techniques for enforcing safety of dynamical systems via invariance constraints include level-set  methods \cite{blanchini2008set}, methods leveraging reachability analysis  \cite{mitchell2000level}, and model-predictive control  methods \cite{gurriet_mpc_invariance}. 

When a nominal but possibly unsafe controller is available, control barrier functions (CBFs), introduced in \cite{ames2014cdcCBFs, ames_tac}, offer a particularly effective tool to enforce safety for control-affine dynamical systems.
Control barrier functions are Lyapunov-like functions which allow for filtering the nominal control inputs in order to provide formal guarantees for safety and have been applied for collision avoidance in multi-robot systems \cite{LiTAC}, adaptive cruise control \cite{ames2014cdcCBFs}, energy aware control \cite{gennaro_persist}, motion planning \cite{mohitCDC2018}, and coverage control \cite{riku_coverage}. CBFs are used in conjunction with quadratic programs (QPs) to compute at each time instant a safe control input. Such approaches are minimally evasive in the sense that the nominal controller's prescribed input is adjusted only when violation of system safety is imminent.

Traditional CBF based approaches ignore the physical dimensions of the system and the emphasis is on the forward invariance of a single point defining the system state. For example, the system state may be the center of mass of a robot or a vehicle. The volume (or extent) of the system is not explicitly incorporated in the formulation of the constraints in the QP. In this paper, we  propose a CBF based method for imposing safety constraints on a control-affine dynamical system with extent. We guarantee safety of the system using a modified CBF constraint which ensures that the extent set always remains within the safe set. We first propose implementing the resulting constraint using a sum-of-squares (SOS) optimization program \cite{sostools}. Since SOS programs can be computationally difficult for high dimensional systems and are only applicable when the safe and extent sets can be represented as polynomials, we next prove that the guarantee on system safety can be retained by considering only a finite set of sampled points on the boundary of the extent set, and we propose a QP based controller using the sampled points. Lastly, the proposed framework is validated both in simulations as well as on a robotic platform.

We note that, in principle, a system's extent can instead be accommodated by shrinking the safe set appropriately \cite{shrink_safe}. However, CBF-based methods rely on characterizing the safe set as a level-set of a function that is known in closed form. Obtaining such a closed form function to accommodate a system's extent is generally not possible. An important exception is when the system's extent and safe set can be represented as balls in the same normed space. This feature is implicitly exploited in, e.g., \cite{li_quads}, which proposes a CBF based method to avoid collisions between teams of quadrotors.  This method requires the volume of each agent to be identical and collision avoidance is then achieved by virtue of the fact that every agent's volume corresponds to identical balls in a normed space. Here, we propose a method that does not require such restrictive assumptions.

This paper is organized as follows: Section~\ref{sec:math_bckg} presents background on control barrier functions. Section~\ref{sec:Ec-CBF} proposes the extent-compatible control barrier function formulation. Section~\ref{sec:Ec-CBF_qp} introduces the QP based controller and presents two methods: a sum-of-squares (SOS) based approach and a sampling based method to guarantee safety. Section~\ref{sec:case_studies} consists of two case studies which implement the proposed framework both in numerical simulations and on a differential drive robot. Section~\ref{sec:conclusion} provides concluding remarks.

%%%%%%%%%%%%%%%%%%%%%%%%%%%%%%%%%%%%%%%%%%%%%%%
%% Section II: Mathematical Background
%%%%%%%%%%%%%%%%%%%%%%%%%%%%%%%%%%%%%%%%%%%%%%%
\section{Mathematical Background}
\label{sec:math_bckg}
Consider a control affine dynamical system
\begin{equation}
\label{eq:system}
  \dot{x}=f(x)+g(x)u \,,
\end{equation}
where $f$ and $g$ are locally Lipschitz continuous, $x \in \mathcal{D} \subset \mathbb{R}^n$ is the state of the system, $\mathcal{D}$ is assumed to be open, and $u\in\mathbb{R}^m$ denotes the control input. Associated with the system~\eqref{eq:system} is a \emph{safe set} $\C \subset \mathcal{D}$, defined as the super zero level set of a continuously differentiable function $h:\mathcal{D} \to \mathbb{R}$, i.e., $\C = \{x \in \mathcal{D} \mid h(x)\geq 0\}$. We call $h$ the \emph{safe function}.

Traditionally, the system \eqref{eq:system} is said to be safe if $x(t) \in \mathcal{C}$ for all $t \geq 0$, given any $x(0) \in \mathcal{C}$. As presented in \cite{CBFs_tutorial}, one can use zeroing control barrier functions (ZCBFs) in order to guarantee forward invariance of a safe set. To that end, a continuous function $\alpha : \mathbb{R} \rightarrow \mathbb{R}$ is \emph{class $\K$} if $\alpha(0) = 0$ and it is strictly increasing.

A continuously differentiable function $h :\mathcal{D}\to\mathbb{R}$ is a \emph{Zeroing Control Barrier Function (ZCBF)} if there exists a locally Lipschitz class $\K$ function $\alpha$ such that for all $x \in \mathcal{D}$,
\begin{equation}
\label{eq:zcbf}
\sup_{u\in \mathbb{R}^{m}} \bigg\{ \frac{\partial h(x)}{\partial x}(f(x)+g(x)u)% + \frac{\partial h(x)}{\partial x}g(x)u
+ \alpha(h(x)) \bigg\} \geq 0 \,.
\end{equation}

In the instance that a safe function $h$ defining a safe set~$\mathcal{C}$ is also a ZCBF, choosing a control input at each state $x$ from the set
\begin{equation}
\label{eq:zcbf_controls}
       U(x) =  
       \left\{ u \in \mathbb{R}^{m} \,\bigg|\, \frac{\partial h(x)}{\partial x}(f(x) + g(x)u) %\frac{\partial h(x)}{\partial x}g(x)u 
       + \alpha(h(x)) \geq 0 \right\}  
\end{equation}
 guarantees that the safe set $\mathcal{C}$ is forward invariant  \cite{CBFs_tutorial}. Specifically, if $x(0) \in \mathcal{C}$ and $U(x)$ as in \eqref{eq:zcbf_controls} is non-empty for all $x \in \mathcal{D}$, then, as shown in  \cite[Theorem 6]{mcbfs}, any continuous feedback controller $u:\mathcal{D}\to\mathbb{R}^m$ such that $u(x) \in U(x) $ for all $x \in \mathcal{D}$ is such that $x(t) \in \mathcal{C}$ for all $t \geq 0$.
 
%%%%%%%%%%%%%%%%%%%%%%%%%%%%%%%%%%%%%%%%%%%%%%%
%% Section III: Extent Compatible Control Barrier Functions
%%%%%%%%%%%%%%%%%%%%%%%%%%%%%%%%%%%%%%%%%%%%%%%

\section{Extent Compatible Control Barrier Functions}
\label{sec:Ec-CBF}
\subsection{Problem Statement}
Given a ZCBF, \cite[Theorem 6]{mcbfs} guarantees that the system state will remain within the safe set $\mathcal{C}$ when control inputs are chosen according to $u(x)\in U(x)$ for all $x \in \mathcal{D}$.  This notion of safety, however, disregards the extent of the system, i.e., when the system state is on the boundary or close to the boundary of the safe set, there still may be some part of the system's physical volume which extends into an unsafe region. To that end, we define a notion of system safety that includes the physical volume of the system. We encapsulate this notion of volume with an extent function.

\begin{definition}[Extent Function]\label{def:ef}
An \emph{extent function} $E:\mathcal{D} \times \mathcal{D} \to \mathbb{R}$ is a continuously differentiable function such that:
\begin{enumerate}
    \item $\E(x) = \{y \in \mathcal{D} \mid E(x,y)\leq 0\}$ is nonempty for all $x\in \mathcal{D}$, 
    \item  $\frac{\partial}{\partial y} E(x,y) \neq 0$ for all $(x,y)$ such that $E(x,y) = 0$, and \item for all $\delta>0$, there exists $\epsilon>0$ such that for all $x\in \mathcal{D}$, if $E(x,y)\leq \epsilon$ then $\inf_{\hat{y}:E(x,\hat{y})=0}\|y-\hat{y}\|\leq \delta$.
    \end{enumerate}
Given an extent function, the set $\mathcal{E}(x)$ above defines the extent of the system when the state of the system is $x \in \mathcal{D}$, and $\partial \mathcal{E}(x) = \{y \in \mathcal{D} \mid E(x,y) = 0 \}$ denotes the extent boundary.
\end{definition}

For example, the extent function $E(x,y)=(x-y)^T(x-y)-r^2$ implies that the extent of the system is a ball of radius $r \in \mathbb{R}_{> 0}$ centered at the system's state $x \in \mathcal{D}$.

We aim to ensure that the extent of the system is contained within the safe set for all time, i.e., $\E(x(t))\subseteq \C$ for all $t\geq 0$ along trajectories of~\eqref{eq:system}.
\begin{prob_stat}
\label{prob_stat}
Given a control affine dynamical system as in \eqref{eq:system} with extent function $E(x,y)$, synthesize a controller which guarantees $\mathcal{E}(x(t)) \subseteq \mathcal{C}$ for all $t \geq 0$ whenever $\mathcal{E}(x(0)) \subset \mathcal{C}$.
\end{prob_stat}

\subsection{Extent-Compatible Control Barrier Function (Ec-CBF)}
Given a system of the form \eqref{eq:system}, we now introduce extent-compatible control barrier functions (Ec-CBFs) which are analogous to ZCBFs, but guarantee that the entire extent set remains within the safe set under the application of a suitable control input.

\begin{definition}[Extent-Compatible Control Barrier Function (Ec-CBF)]
\label{def:ec_cbf}
A continuously differentiable function $h$ is an \emph{extent-compatible control barrier function} (Ec-CBF) for the system \eqref{eq:system} with extent function $E$ if $\frac{\partial}{\partial x}h(x)\neq 0$ for all $x$ such that $h(x)=0$ and there exists locally Lipschitz class~$\mathcal{K}$ functions $\alpha_1$ and $\alpha_2$ such that for all $x\in \mathcal{D}$ with $\mathcal{E}(x)\subseteq \mathcal{C}$ and for all $y \in \mathcal{C}$, defining
\begin{multline*}
    \mathcal{M}(x,y,u):=\\
    \frac{\partial E(x,y)}{\partial x}(f(x)+g(x)u)+\alpha_1(E(x,y)) + \alpha_2(h(y)) \, ,
\end{multline*}
 it holds that $\sup_{u\in\mathbb{R}^m} \{\mathcal{M}(x,y,u)\}\geq 0$.
\end{definition}

Given an Ec-CBF $h$, for all $x \in \mathcal{D}$, define the set
\begin{align}
    \label{eq:controls_Ec-CBF}
    \mathcal{U}(x) = \{u\in\mathbb{R}^m\mid \mathcal{M}(x,y,u)\geq 0\text{ for all }y\in\mathcal{C}\} \,.
    %\{u\in\mathbb{R}^m \mid \frac{\partial E(x,y)}{\partial x}f(x) + \frac{\partial E(x,y)}{\partial x}g(x)u \geq \nonumber \\ 
    %-\alpha_1(E(x,y)) - \alpha_2(h(y))  \text{ for all } y\in \mathcal{C}\} \, .
\end{align}
Assuming that the extent of the system initially begins inside the safe region, choosing a control input $u(x) \in \mathcal{U}(x)$ at any given state $x \in \mathcal{D}$ guarantees that $\mathcal{E}(x(t)) \subseteq \mathcal{C}$ for all $t \geq 0$, %solving Problem \ref{prob_stat}, 
as formalized in the following theorem.

\begin{thm}
\label{thm:Ec-CBF}
Consider system~\eqref{eq:system} with initial state $x(0)$, an extent function $E$, and an Ec-CBF $h$ with associated safe set $\mathcal{C}=\{x\in \mathcal{D} \mid h(x)\geq 0\} \subset \mathcal{D}$. If $\mathcal{E}(x(0)) \subset \mathcal{C}$, then any continuous feedback controller $u:\mathcal{D}\to\mathbb{R}^m$ such that $u(x) \in \mathcal{U}(x)$ for all $x \in \mathcal{D}$ guarantees that $\mathcal{E}(x(t)) \subseteq \mathcal{C}$ for all $t \geq 0$.
\end{thm}

\begin{proof}
% Comparison lemma version
Suppose by contradiction that the assumptions of the theorem hold but there exists a time $t' > 0$ such that $\mathcal{E}(x(t')) \not\subseteq \mathcal{C}$, i.e., there exists a point $y' \in \mathcal{D}$ such that $E(x(t'), y') \leq 0$ and $h(y') < 0$. If $E(x(t'), y') =0$, then because $\frac{\partial}{\partial y} E(x(t'),y') \neq 0$ by the definition of extent function, in any neighborhood of $y'$ there exists $y''$ such that $E(x(t'), y'') < 0$ and $h(y'')<0$, hence we assume $E(x(t'), y')<0$ without loss of generality. A continuity argument then implies there exists $0<t'''\leq t'$ and $y'''\in \mathcal{D}$ such that $h(y''')=0$ and $E(x(t'''), y''') < 0$.

Let $w(t) = E(x(t), y''')$.  Then, for all $t\geq 0$,
\begin{align*}
\dot{w}(t) 
   &= \frac{\partial E(x(t), y''')}{\partial x}\big(f(x(t))+g(x(t))u(x(t))\big)%  + \frac{\partial E(x(t),y''')}{\partial x}g(x(t))u(x(t)) \\
   \\
            &\geq -\alpha_{1}(w(t)) -\alpha_{2}(h(y'''))\\
            & \geq -\alpha_{1}(w(t)) \,,
\end{align*}
where the first inequality holds since $h$ is an Ec-CBF and the second inequality follows because $-\alpha_{2}(h(y''')) = 0$.

Now, consider the initial value problem $\dot{\eta}(t) = -\alpha_1(\eta(t))$ with $\eta(0) = w(0)$.   Since $\mathcal{E}(x(0)) \subset \mathcal{C}$ and  $\frac{\partial }{\partial y}h (y''')\neq 0$, it must hold that $\eta(0) = w(0) \geq 0$. %From the fact in~\eqref{eq:scalar_sys1}--\eqref{eq:scalar_sys3} that $\dot{w}(t) \geq -\alpha_1(w(t))$ together with $\eta(t) \geq  0$ for all $t \geq 0$, 
The comparison lemma~\cite[Lemma 3.4]{khalil2002nonlinear} then implies $w(t) \geq \eta(t) \geq 0$ for all $t \geq 0$.
%$\dot{\eta}(t) = -\alpha(\eta(t))$ and $\eta(0) = w(0) > 0$. From Lemma 4.4 in \cite{khalil2002nonlinear}, the solution to the IVP is given by $\eta(t) > 0$ for all $t \geq 0$. 
This contradicts the claim that $w(t''') = E(x(t'''), y''') < 0$. Hence, $\mathcal{E}(x(t)) \subseteq \mathcal{C}$ for all $t \geq 0$. 
\end{proof}

\begin{rmk}\label{rmk:thmebf}
From the proof of Theorem~\ref{thm:Ec-CBF}, we observe that $\sup_u \{\mathcal{ M}(x,y,u)\} \geq 0$ from Definition \ref{def:ec_cbf} only has to hold for  $y \in \mathcal{D}$ such that $E(x,y)$ is in a neighborhood of $0$. 
Part 3) in Definition~\ref{def:ef} ensures that for a small enough neighborhood, such $y$ must be arbitrarily close to the the boundary $\partial \mathcal{E}(x).$% this can be guaranteed by considering proximity can be characterized by $|E(x,y)| < \epsilon$ for some $\epsilon > 0$. 
\end{rmk}

Next, we propose an optimization scheme for selecting inputs from the set $\mathcal{U}(x)$ to guarantee safety while minimally deviating from some prescribed nominal controller.

%%%%%%%%%%%%%%%%%%%%%%%%%%%%%%%%%%%%%%%%%%%%%%%%
%% Subsection IV: Minimally Invasive Quadratic Program Based Controller
%%%%%%%%%%%%%%%%%%%%%%%%%%%%%%%%%%%%%%%%%%%%%%%%
\section{Minimally Invasive Quadratic Program Controller}
\label{sec:Ec-CBF_qp}
Consider the scenario where a system designer would like to employ some nominal feedback control policy $k : \mathcal{D} \rightarrow \mathbb{R}^m$ on the system \eqref{eq:system} with extent. When safety of the system under the control policy $k$ cannot be verified a priori, we propose incorporating \eqref{eq:controls_Ec-CBF} as a constraint at runtime to obtain a minimally invasive quadratic program (QP) based controller using a Ec-CBF, similar to the technique proposed in \cite{ames2014cdcCBFs, ames_tac} for ZCBFs. This procedure leads to a control law which ensures that the extent set of the system is contained within the safe set $\mathcal{C}$ for all $t \geq 0$, given that $\mathcal{E}(x(0)) \subset \mathcal{C}$. In particular, we propose a quadratic program based controller of the form
\begin{equation}
    \label{eq:qp}
    u_\text{QP}(x)= \argmin_{u\in\mathcal{U}(x)}  \:\; \|u-k(x)\|_2^2 \,.
    %\text{s.t.}& \quad\text{\eqref{eq:Ec-CBF} holds $\forall y\in \C$.} %\nonumber
\end{equation}

The above QP is minimally invasive in the sense that it guarantees safety of the system including its extent, while following the nominal input $k$ with minimal deviation. However, for fixed $x$, $\mathcal{U}(x)$ is defined from \eqref{eq:controls_Ec-CBF} and requires a given inequality to hold for all $y$, leading to an infinite number of linear constraints on $u$. In the remainder of this section, we present two approaches to obtaining computationally tractable programs that retain safety guarantees. 

\subsection{Optimization Over Sum-of-Squares Polynomials}
We first propose recasting \eqref{eq:qp} as a sum-of-squares (SOS) optimization problem in the independent variable $y$. 

\begin{definition}[Sum-of-Squares (SOS) Polynomials]
\label{def:sos}
Let $\mathbb{R}[y]$ be the set of all polynomials in $y \in \mathbb{R}^n$. Then
\begin{equation*}
    \Sigma[y] = \left\{ s(y) \in \mathbb{R}[y] \,\bigg{\vert}\, s(y) = \sum\limits_{i=1}^{\ell} p_{i}(y)^2,\, p_{i}(y) \in \mathbb{R}[y]\right\}
\end{equation*}
where $\ell \in \mathbb{N}$, is the set of SOS polynomials. Note that if $s(y) \in \Sigma[y]$, then $s(y) \geq 0$ for all $y \in \mathbb{R}^n$.
\end{definition}

\begin{thm}
\label{prop:sos}
Consider system~\eqref{eq:system} with initial state $x(0)$, an extent function $E$, and an Ec-CBF $h$ with associated safe set $\mathcal{C}=\{x\in \mathcal{D} \mid h(x)\geq 0\} \subset \mathcal{D}$. Further assume $E(x,y)$, $h(y)$, $\alpha_1(E(x,y))$ and $\alpha_2(h(y))$ are polynomial in $y$ for any fixed $x$. Let $k : \mathcal{D} \rightarrow \mathbb{R}$ be a continuous nominal controller and suppose $\mathcal{E}(x(0)) \subset \mathcal{C}$. If the set
\begin{align*}
\widetilde{\mathcal{U}}(x) = \biggl\{ u \in \R^m \, \biggr\lvert \, &\frac{\partial E(x,y)}{\partial x}(f(x)+g(x)u)+\alpha_1(E(x,y))\\   &+\alpha_2(h(y))-s(y)h(y)\in\Sigma[y] \\
    &\text { for some $s(y) \in \Sigma[y]$} \biggr\}
\end{align*} is non-empty for all $x \in \mathcal{D}$, then the solution $x(t)$ of  system~\eqref{eq:system} with 
\begin{equation}
u(x)= u_\text{SOS}(x) := \argmin_{u\in \widetilde{\mathcal{U}}(x)} \|u-k(x)\|_2^2  \label{eq:SOSobj}
\end{equation}
is such that $\mathcal{E}(x(t)) \subseteq \mathcal{C}$ for all $t \geq 0$.
\end{thm}
\begin{proof}
For each $x$, the optimization problem~\eqref{eq:SOSobj} is feasible by hypothesis, and the fact that $u \in \mathcal{\widetilde U}(x)$ implies 
\begin{multline*}
\frac{\partial E(x,y)}{\partial x}(f(x)+g(x)u)+\alpha_{1}(E(x,y))\\
+ \alpha_{2}(h(y)) - s(y)h(y) \geq 0 
\end{multline*}
for all $y \in \mathbb{R}^n$, since the left hand side of the inequality is required to be an SOS polynomial. %We know that $s(y)$ is an SOS polynomial and hence $s(y) \geq 0$ for all $y \in \mathbb{R}^n$.
Next, observe that $s(y)h(y)\geq 0$ for all $y\in\mathcal{C}$ since $s(y)$ is a SOS polynomial and $h(y) \geq 0$. Hence the requirement that $u \in \mathcal{\widetilde{U}}(x)$ implies
\begin{multline*}
    \frac{\partial E(x,y)}{\partial x}(f(x)+g(x)u) +\alpha_{1}(E(x,y))+\alpha_{2}(h(y)) \geq 0 \,
\end{multline*}
for all $y\in\mathcal{C}$, i.e., $u\in \mathcal{U}(x)$ as defined in \eqref{eq:controls_Ec-CBF}. In addition, since for all $x \in \mathcal{D}$, the constraint in $\widetilde{\mathcal{U}}(x)$
is quasi-convex in $u$, $\|u\|$ is quasi-convex in $u$, and $k$ is continuous in $x$, using~\cite[Proposition 8]{mcbfs}, we conclude that the controller is continuous.
Hence, from Theorem~\ref{thm:Ec-CBF}, $\mathcal{E}(x(t)) \subseteq \mathcal{C}$ for all $t \geq 0$.
\end{proof}

Proposition \ref{prop:sos} leads to an online SOS based controller by solving \eqref{eq:SOSobj} at each current state $x$. To implement this controller in, e.g., SOSTOOLS~\cite{sostools}, the degree of the SOS decision polynomial $s(y)$ in the constraint defining $\widetilde{\mathcal{U}}(x)$ is fixed a priori. In addition, the quadratic cost in \eqref{eq:SOSobj} is recast in epigraph form to obtain a new linear cost; by exploiting the \textit{Schur complement} test for positive semi-definiteness  \cite{vandenberghe1996semidefinite}, the initial formulation \eqref{eq:SOSobj} can be equivalently rewritten as
\begin{equation*}
   u_\text{SOS}(x) =  \argmin_{u \in \widetilde{\mathcal{U}}(x)} \min_{\delta \in \Delta(u, x)} \delta \,, \label{eq:SOSobjdelta} 
\end{equation*}
with 
\begin{multline*}
    \Delta(u, x) =  \\ \left\{ \delta \in \mathbb{R} \bigm|  
    \begin{bmatrix}
    I & u\\
    u^T & \delta+2k(x)^Tu-k(x)^Tk(x) \end{bmatrix} \succeq 0
    \right\} \,.
\end{multline*}

The above SOS approach allows us to adopt a tractable method to guarantee safety for the system. However, with increasing system dimensionality, SOS programs are known to become computationally difficult. Hence, we next propose an alternative efficient sampling based approach.

\subsection{A Sampling Based Approach to Set Invariance with Extent}
\label{sec:sampling_method_Ec-CBF}
In this subsection, we propose an alternative relaxation of \eqref{eq:qp} which retains the computational advantages of the original QP formulation for ZCBFs. The intuition is to enforce the constraint in~\eqref{eq:controls_Ec-CBF} on the boundary of the extent set, but only for a finite number of sampled points.

To obtain a finite number of sampled points, we discretize the boundary of the extent set $\partial \mathcal{E}(x)$.  The following theorem guarantees safety by introducing a constraint on the control input for each sampled point.
\begin{thm}
\label{thm:sampling}
Consider system~\eqref{eq:system} with initial state $x(0)$, an extent function $E$, and an Ec-CBF $h$ with associated safe set $\mathcal{C}=\{x\in \mathcal{D} \mid h(x)\geq 0\} \subset \mathcal{D}$. Further assume that the domain $\mathcal{D}$ is bounded, $\mathcal{E}(x)$ is bounded for all $x \in \mathcal{D}$, 
let $M>0$ be a bound on the magnitude of the control input, and for some $\tau>0$  let $\partial \mathcal{E}_{\tau}(x) \subset \partial \mathcal{E}(x)$ be a finite set such that for all $y \in \partial \mathcal{E}(x)$, it holds that $\min_{\widetilde{y} \in \partial \mathcal{E}_{\tau}(x) } \|\widetilde{y}-y\| \leq \tau/2$.
Additionally, let
\begin{align}
    \label{eq:defA}A &\geq \sup_{y \in \mathcal{D}} \norm{\frac{\partial}{\partial y} h(y)} \,, \\
    B &\geq \sup_{x,y \in \mathcal{D},  \, \|u\| < M} \norm{  \frac{\partial E (x,y)}{\partial x} (f(x) + g(x)u)} \,. \label{eq:defB}
\end{align}

Consider the set
\begin{multline*}
    \mathcal{\widehat U}(x) = \biggl\{ u \in \mathbb{R}^m \bigm\lvert \|u\| \leq M \text{ and } \\
     \frac{\partial E(x, y^{*})}{\partial x} (f(x) + g(x)u) + \gamma \cdot h(y^{*}) \geq (B + \gamma A) \tau \\
    \text { holds for all $y^*\in \partial \mathcal{E}_{\tau}(x)$} \biggr\}
\end{multline*}
where $\gamma > 0$ is a constant. Let $k : \mathcal{D} \rightarrow \mathbb{R}$ be a continuous nominal control input. If for all $x \in \mathcal D$ the set $\mathcal{\widehat U}(x)$ is non-empty and $\mathcal{E}(x(0))\subset \mathcal{C}$, then the solution $x(t)$ to the system~\eqref{eq:system} with $u = u_{\text{sampled}}(x)$, where
\begin{align}
    \label{eq:sampled_qp}
    u_{\text{sampled}}(x)= \argmin_{u\in \mathcal{\widehat U}(x)} \|u-k(x)\|_2^2 \,,
\end{align}
is such that $\mathcal{E}(x(t)) \subseteq \mathcal{C}$ for all $t \geq 0$. 
Moreover, the controller $u_{\text{sampled}}(x)$ is continuous with respect to $x$ for all $x \in \mathcal{D}$.
\end{thm}
\begin{proof}
Introduce the open set
\begin{equation*}
\mathcal{B} = \bigcup_{y^*\in \partial \mathcal{E}_\tau(x)} B^o_\frac{3\tau}{4}(y^*) \,,
\end{equation*}
where $B^o_\frac{3\tau}{4}(y^*)$ denotes an open ball with radius $\frac{3\tau}{4}$ centered around $y^*$. Choose $\epsilon > 0$ such that for all $x$, 
\begin{equation*}
\mathcal{\bar E}(x) = \{y \in \mathcal{D} \mid |E(x,y)| \leq \epsilon\}\subseteq \mathcal{B} \,.
\end{equation*}
Such a choice of $\epsilon$ is possible due to part 3 of the definition of the extent function (Definition~\ref{def:ef}). Clearly $\partial \mathcal{E}_{\tau}(x) \subset \partial \mathcal{E}(x) \subset \mathcal{\bar E}(x)$. 
From the mean value theorem, for all $y \in \mathcal{\bar E}(x)$ and $y^{*} = \argmin_{\bar{y} \in \partial \mathcal{E}_\tau(x)} \norm{\bar{y} - y}$, it follows that
\begin{equation}
    h(y^*) - h(y)  \leq A \| y^* -y \| \, , \label{eq:samph}
\end{equation}
and
\begin{equation} 
\left( \frac{\partial E (x,y^*)}{\partial x} - \frac{\partial E (x,y)}{\partial x}\right) (f(x) + g(x)u)  \leq B \|y^* - y\|  \,  \label{eq:sampB}
\end{equation}
whenever $\|u\|\leq M$. Now, observe that $\norm{y^* -y} \leq 3\tau/4$. Multiplying~\eqref{eq:samph} with $-\gamma$ and~\eqref{eq:sampB} with $-1$ and then adding the inequalities yields
\begin{align*}
%\label{eq:samp_const_proof}
&\frac{\partial E (x,y)}{\partial x} (f(x) + g(x)u) + \gamma \cdot h(y)\\
&\geq  \frac{\partial E (x,y^*)}{\partial x} (f(x) + g(x)u) + \gamma \cdot h(y^*) - (B + \gamma A) 3\tau/4 \\
&\geq  (B + \gamma A) \tau/4 \,
\end{align*}
for any $u\in \mathcal{\widehat U}(x)$ where the last inequality follows from the definition of $\mathcal{\widehat U}(x)$. 

Choosing $\alpha_1(s) = (B+\gamma A)\tau/(4\epsilon) s$ and $\alpha_2(s)=\gamma s$, for any $u\in \mathcal{\widehat U}(x)$,  we have $\mathcal{M}(x,y,u)\geq 0$ for all $y$ such that $|E(x,y)|\leq \epsilon$ where $\mathcal{M}$ is as in Definition \ref{def:ec_cbf}, and therefore, in accordance with Remark~\ref{rmk:thmebf}, $h$ is a Ec-CBF and $u_\text{sampled}(x)\in \mathcal{U}(x)$ for all $x\in \mathcal{D}$. 

Since for all $y^{*} \in \partial \mathcal{E}_{\tau}(x)$, in the definition of $\widehat{\mathcal{U}}(x)$, the constraint 
\begin{equation*}
\frac{\partial E(x, y^{*})}{\partial x} (f(x) + g(x)u) + \gamma \cdot h(y^{*}) \geq (B + \gamma A) \tau
\end{equation*}
is quasi-convex in $u$, $\|u\|$ is quasi-convex in $u$, and $k$ is continuous in $x$, using~\cite[Proposition 8]{mcbfs}, we conclude that the controller $u_{\text{sampled}}(x)$ is continuous. Thus, from Theorem~\ref{thm:Ec-CBF},
the extent set $\mathcal{E}(x)$ is contained within the safe set for all $t \geq 0$ and for all $x \in \mathcal{D}$ such that $E(x(0))\subseteq \mathcal{C}$; that is, $\mathcal{E}(x(t)) \subseteq \mathcal{C}$ for all $t \geq 0$.
\end{proof}

The computational complexity of \eqref{eq:sampled_qp} scales with the cardinality of $\partial \E_{\tau}(x)$; however, as we show in the following example, choosing $\tau$ to be large can result in unwanted conservatism.  This  indicates an underlying trade-off between the performance of \eqref{eq:sampled_qp}, which scales inversely with $\tau$, and computational resources required to compute \eqref{eq:sampled_qp} online.

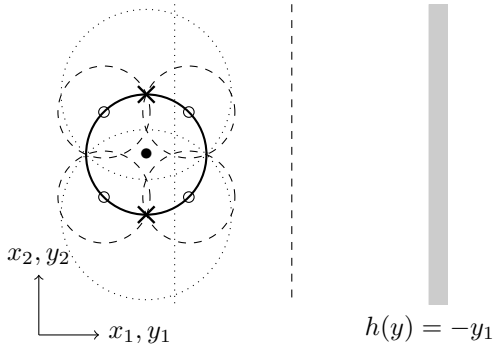
\begin{figure}
    \centering
    \begin{tikzpicture}[scale=0.8]
    \draw[thick] (0,2.5)  -- (0,-2.5) node[below] {$h(y)= -y_1$};
    \fill[black!20] (0,2.5) rectangle (0.3,-2.5); 
    
    \draw[dotted] ({-3*sqrt(2)}, 2.5) -- ({-3*sqrt(2)}, -2.5);
    
    \draw[dashed] (-2.2961, 2.5) -- (-2.2961, -2.5);

     \begin{scope}[shift={(-3.3,0)}]     
     % Large tau
     \draw[dotted] (-{sqrt(2)},1) circle ({sqrt(2)});
     \node at (-{sqrt(2)},1) {\Cross};
     \draw[dotted] (-{sqrt(2)},-1) circle ({sqrt(2)});
     \node at (-{sqrt(2)},-1) {\Cross};

     % Smaller tau
     \draw[dashed] ({-sqrt(2)+1/sqrt(2)}, {1/sqrt(2)}) circle (0.765);
     \node[draw, circle, inner sep=0pt,minimum size=4pt] at ({-sqrt(2)+1/sqrt(2)}, {1/sqrt(2)}) {};
     \draw[dashed]  ({-sqrt(2)+1/sqrt(2)}, {-1/sqrt(2)}) circle (0.765);
     \node[draw, circle, inner sep=0pt,minimum size=4pt] at ({-sqrt(2)+1/sqrt(2)}, {-1/sqrt(2)}) {};

     \draw[thick] (-{sqrt(2)},0) circle (1);
     \node at (-{sqrt(2)},0) {\textbullet};
     \end{scope}
     
     \begin{scope}[shift={(-4.71,0)}]
     \draw[dashed] ({-sqrt(2)+1/sqrt(2)}, {1/sqrt(2)}) circle (0.765);
     \node[draw, circle, inner sep=0pt,minimum size=4pt] at ({-sqrt(2)+1/sqrt(2)}, {1/sqrt(2)}) {};
     \draw[dashed]  ({-sqrt(2)+1/sqrt(2)}, {-1/sqrt(2)}) circle (0.765);
     \node[draw, circle, inner sep=0pt,minimum size=4pt] at ({-sqrt(2)+1/sqrt(2)}, {-1/sqrt(2)}) {};
     \end{scope}
    \draw[->] (-6.5,-3) -- +(0,1) node[above] {$x_2, y_2$};
    \draw[->] (-6.5,-3) -- +(1,0)  node[right] {$x_1, y_1$};
    
    \end{tikzpicture}
    
    \caption{Example of the different discretizations in Example~\ref{ex:sampling}. The black dot is the system and the solid circle its extent. In the case with two discretization points, the points are marked with crosses, and the dotted circles show all points with the distance $\tau$ away from the discretion points, using the $2$-norm. For the case of four discretion points, the points marked with circles and the distances with dashed circles. The dotted and dashed lines show how close the discretizations points can be to the barrier and still satisfy the inequality, for $\tau=2\sqrt{2}$ and $\tau=2\sqrt{2-\sqrt{2}}$ respectively.}
    \label{fig:sampledset}
\end{figure}

\begin{example}\label{ex:sampling}
Consider the system $\dot{x} = u$, where $x = \begin{bmatrix} x_1 & x_2 \end{bmatrix}^{T} \in \R^2$ is the system state and $u \in \R^2$ is a bounded control input such that $\norm{u}_{2} \leq M$ where $M = 1$.
%$u_1,\,u_2 \leq 1$
We take $E(x,y) = (x_1-y_1)^2 + (x_2-y_2)^2 - 1$ and encode the safety constraint with the safe function $h(y) = -y_1$; equivalently, the system is said to be safe if the extent set, 
here taken to be a circle of radius one, lies entirely in the closed left-half plane. This problem setting is depicted in Fig.~\ref{fig:sampledset}. %For this specific example, 
We take  Lipschitz constants $A = 1$ and $B = 2$ satisfying \eqref{eq:defA} and \eqref{eq:defB}.

We first construct the sampling based controller \eqref{eq:sampled_qp} with two samples; in this case, $\partial\mathcal{E}_{\tau} (x) = \{(x_1, x_2 \pm 1)\}$ such that $\tau = 2\sqrt{2}$.
The controller \eqref{eq:sampled_qp} then has two constraints, namely,
\begin{equation}\label{exconst}
    \pm 2u_2 \geq  (B + \gamma A) \tau  - \gamma h(y^*) =  (B + \gamma A) \tau  + \gamma  x_1\,,
\end{equation}
which can only be satisfied when $x_1 \leq - (\frac{B}{\gamma} +  A) \tau$. Observe that $\gamma$ plays a key role in the behavior of the system. A high value of $\gamma$ will render $-(\frac{B}{\gamma} + A)\tau \approx - A \tau$. Thus, the term $B$, which impacts the effect of the dynamics in the extent-barrier constraint, vanishes. This is in line with the intuition one observes in the traditional barrier function case, where a high slope of the class-$\mathcal{K}$ function will allow the system to quickly approach the boundary of the safe set.
Interestingly, the control input $u_1$ is not restricted by \eqref{exconst}; however, it has still to satisfy $\norm{u}_2 \leq M$ and the resulting condition on $x_1$ is quite conservative.

To obtain a less conservative solution,
one must increase the number of sampling points, and hence decrease the value of $\tau$. By doing so, it is more likely that the matrix
\begin{equation*}
    S =
    \begin{bmatrix} 
     \frac{\partial E (x,y^{(1)})}{\partial x}^T &  \frac{\partial E (x,y^{(2)})}{\partial x}^T & \ldots &  \frac{\partial E (x,y^{(n)})}{\partial x}^T 
    \end{bmatrix} \,,
\end{equation*}
where $y^{(1)},\, y^{(2)},\, \ldots,\, y^{(n)}$ denote the samples, has full rank. If the matrix $S$ has full rank, both the control inputs $u_1$ and $u_2$ will be constrained.

We demonstrate this property by considering a sampling based controller which uses four samples; that is, we take $\partial \mathcal{E}_\tau(x) = \{ (x_1 \pm 1/\sqrt{2}, x_2 \pm 1/\sqrt{2}), (x_1 \pm 1/\sqrt{2}, x_2 \mp 1/\sqrt{2})\} $, such that $\tau = 2\sqrt{2 - \sqrt{2}}$.  In this case the matrix $S$ has full rank and the controller \eqref{eq:sampled_qp} has four constraints,
\begin{align}
\label{exconst1}
\sqrt{2}(-u_1 \pm  u_2) - \gamma\left(x_1 - \sqrt{2}\right) &\geq (B+\gamma A)\tau \,, \\
\label{exconst2}
\sqrt{2}(u_1 \pm  u_2) - \gamma\left(x_1 + \sqrt{2}\right) &\geq (B+\gamma A)\tau \,.
\end{align}
Importantly, this new instantiation has reduced conservatism, i.e., solutions to \eqref{exconst1}--\eqref{exconst2} exist for all states such that $x_1 \leq -(\frac{B}{\gamma} + A)\tau$, but now with a smaller $\tau$. Moreover, in the instance that $x_1 = -(\frac{B}{\gamma} + A)\tau$, the only feasible solution is $(u_1, u_2) = (-\gamma, 0)$ with $\gamma \leq M$, which will effectively steer the system away from the barrier. Again, observe that $\gamma$ plays a key role in the behavior of the system. A higher value of $\gamma$ will allow for the system to get closer to the barrier, but once the system is close to the boundary, a more aggressive control action,  i.e., $u_1 = \gamma$ is applied.
\end{example}

This simple example is designed to illustrate the importance of the sampling choice, but, as discussed in the introduction, an alternative approach to enforcing safety for this specific example is to artificially contract the safe set by shifting the unsafe region to the left by one unit to account for the extent set. However, as we will see in the following case studies, generating the appropriate amount of contraction for complicated polynomial safe sets is not straightforward.
\section{CASE STUDIES}
\label{sec:case_studies}
In this section, we present two implementations\footnote{Code -- \url{https://github.com/gtfactslab/ExtentCBF}} of the proposed framework on a system with unicycle dynamics
\begin{align*}
    \dot{x}_{1} &= v \cdot \cos(\phi) \,, \\
    \dot{x}_{2} &= v \cdot \sin(\phi) \,, \\
    \dot{\phi} &= \omega \,,
\end{align*}
where $x_{1} \in \mathbb{R}$, $x_{2} \in \mathbb{R}$ are the position coordinates of the robot, $\phi \in [-\pi, \pi)$ is the orientation, and $v \in \mathbb{R}$ and $\omega \in \mathbb{R}$ are the linear velocity input and angular velocity input. Define $\widehat{x} = \begin{bmatrix} x_1 & x_2 \end{bmatrix}^{T}$, and $x = \begin{bmatrix} x_1 & x_2 & \phi \end{bmatrix}^{T}$.

Case study 1 presents a MATLAB simulation of the sum-of-squares (SOS) program based controller whereas Case study 2 is an implementation conducted on the Robotarium \cite{robotarium} (a multi-robot test bed at Georgia Tech) which uses the sampling based controller on a differential drive mobile robot. 

\subsection{Case Study 1: SOS Based Approach}
Define the safe set $\mathcal{C} = \{\widehat{x} \in \mathbb{R}^2 \mid h(\widehat{x}) \geq 0 \}$ where $h(x) = 1 - \widehat{x}^{T}\widehat{x}$. The extent set $\mathcal E(x)$ is a fourth-order superellipse described by the extent function 
\begin{multline*}
    E(x,\, y) = (1.5)^4 (\Delta_1\cos{(\phi)} + \Delta_2\sin{(\phi)})^4 + \\ (2)^4 (-\Delta_1\sin{(\phi)} + \Delta_2\cos{(\phi)})^4 - (0.2)^4 
\end{multline*}
where $\Delta_i = (x_i - y_i)$ for $i = 1,2$.
In Fig~\ref{fig:cs1}, the system trajectory under the nominal control input when no barrier functions are enforced is characterized by the orange dotted line. As can be seen, the system violates the safe set since there are no safety constraints enforced on the system's motion. The trajectory when the system's input is filtered with %is subjected to the nominal input in the presence of 
a traditional ZCBF based controller, with \eqref{eq:zcbf_controls}
 as the constraint, is characterized by the purple line. Observe that the state $x$ remains within the safe set but a part of the extent set moves out of the safe set. The system trajectory when the system's input is filtered using the SOS based Ec-CBF controller in Section~\ref{sec:Ec-CBF_qp} is characterized by the solid line. Observe that the state $x$ as well as the extent set $\mathcal{E}(x)$ are always contained within the safe set.
\begin{figure}
    \centering
    \input{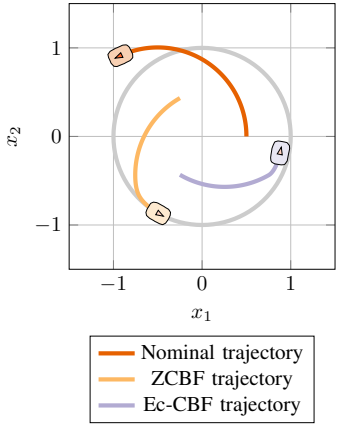}

    \caption{Comparison of the system trajectory without a CBF controller, with a traditional ZCBF controller (QP with \eqref{eq:zcbf_controls} as a constraint), and the proposed Ec-CBF controller (Section~\ref{sec:Ec-CBF_qp}) that ensures the extent set is always contained within the safe set. %In the case of the traditional CBF, the extent set is not contained within the safe set whereas with our proposed approach, the extent set is contained within the safe set for all $t \geq 0$.
    }
    \label{fig:cs1}
    
\end{figure}

\subsection{Case Study 2: Sampling Based Controller}
In this implementation, we consider a second order ellipsoidal extent set of the form $E(x,y) = (\widehat{x}-y)^{T}R(-\phi)^{T} P R(-\phi) (\widehat{x}-y) - 1^2$ where $P = \text{diag}({1.5^{-2}, 1^{-2}})$ is the matrix representing the dimensions of the ellipse. We consider a differential drive mobile robot with unicycle kinematics as defined before. %In this case study, we utilize the sampling based controller as in \eqref{eq:sampled_qp} and the constraint from Theorem~\ref{thm:sampling}.
The safe set $\mathcal{C}$ is characterized by a weighted polar $L_p$ function \cite{mohitCDC201_Lp}, with parameters $\sigma = (1, 0.4)$, $\theta_{k} = \frac{\pi}{2}$, $\kappa = \frac{\theta_{k}}{2 \sigma(1)}$, $p = 4$, $\theta_{0} = \mbox{sign}(\kappa) \cdot \frac{\pi}{2}$. %, to capture complex shapes. 
%The objective is to ensure that the robot along with its ellipsoidal extent set $\mathcal{E}(x)$ is contained within the $L_p$ safe set for all $t \geq 0$. 
The controller as in \eqref{eq:sampled_qp} is used, and from Theorem~\ref{thm:sampling}, safety of the robot is guaranteed. The extent set boundary was sampled with $200$ points. Thus, the QP \eqref{eq:sampled_qp} was solved with $200$ constraints point-wise in time. The time for solving the QP at each time step was between 10ms to 15ms. Figure~\ref{fig:case_study_2} shows a still shot of the implementation conducted on the Robotarium\footnote{Video of experiment -- \url{https://youtu.be/uvrBvclD8ME}}. As can be seen, the robot and its extent set are contained within the $L_p$ safe set for all $t \geq 0$.
\begin{figure}
\begin{center}
    \includegraphics[width=0.85\columnwidth]{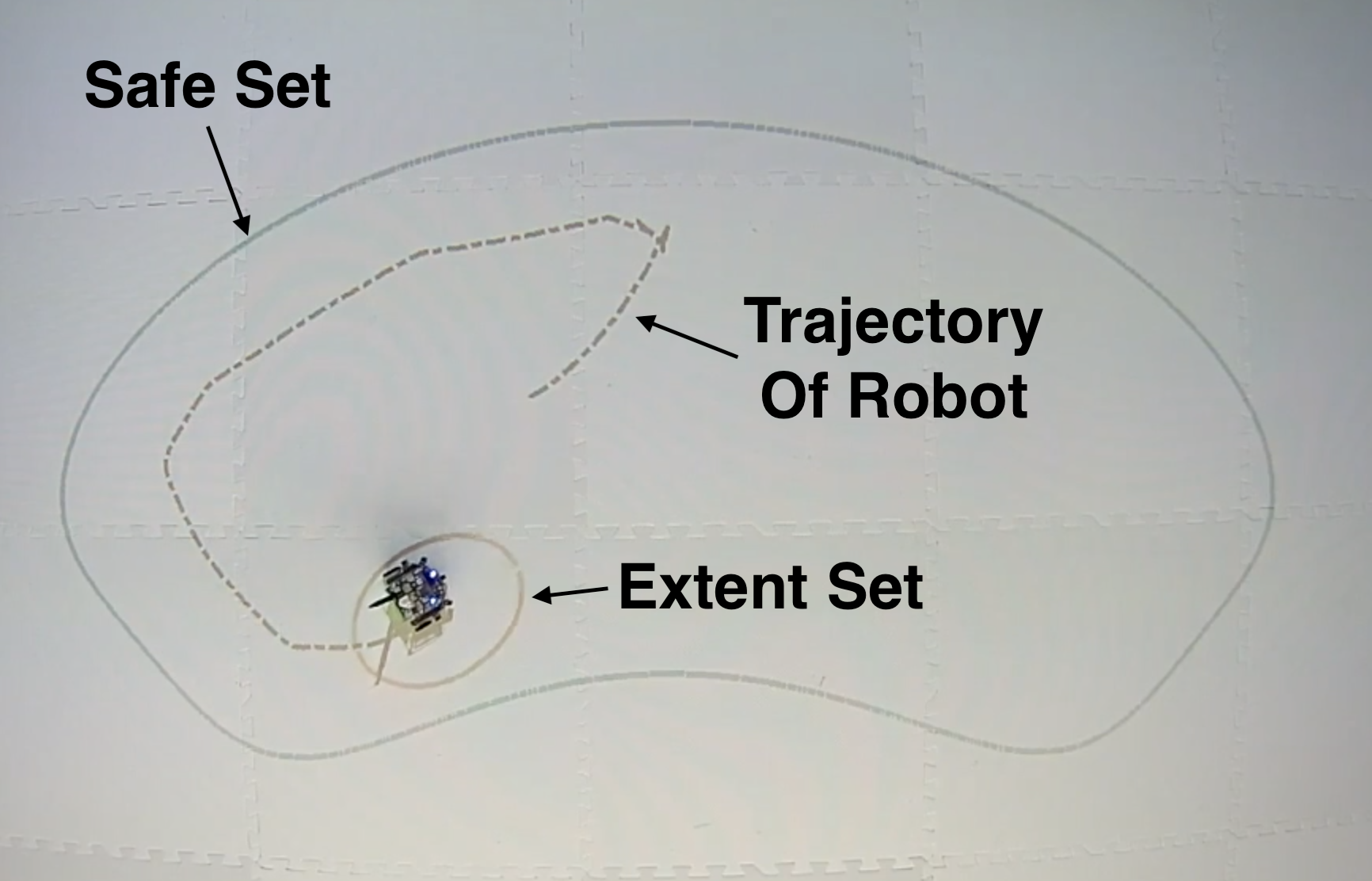}
      \caption{Robotarium implementation of the proposed extent compatible control barrier functions (Ec-CBF) framework. In particular, the sampling based controller \eqref{eq:sampled_qp} is implemented. The safe set is the $L_p$ function (colored in blue). The robot is encapsulated by the ellipsoidal extent set (colored in yellow). The proposed sampling based controller \eqref{eq:sampled_qp} ensures that the extent set is always contained within the $L_p$ safe set.}
\label{fig:case_study_2}
\end{center}
\end{figure}

\section{CONCLUDING REMARKS} 
\label{sec:conclusion}
This paper presents a barrier function based method for ensuring the safety of a control-affine dynamical system that incorporates its physical volume, i.e., its extent. The first resulting controller design relies on a sum-of-squares based optimization program. Since sum-of-squares programs can be computationally difficult, a sampling based method is also presented. The proposed sampling based controller is shown to retain the guarantee on safety of the system. %, and is also shown to be continuous. 
Simulation and robotic implementation results are also provided.
\bibliographystyle{ieeetr}
\bibliography{bib.bib}
\end{document}